\documentclass[11pt,letterpaper]{article}
\usepackage[T1]{fontenc} 
\usepackage[utf8]{inputenc}
\usepackage{microtype}
\usepackage{fullpage}
\usepackage{graphicx} 
\usepackage{amsthm}
\usepackage{amsmath}
\usepackage{amssymb}
\usepackage{mathtools}
\usepackage{color}
\usepackage{thm-restate,color}
\usepackage[table,svgnames]{xcolor}
\usepackage[colorlinks=true,linkcolor=black,citecolor=MidnightBlue,urlcolor=MidnightBlue]{hyperref}
\usepackage{cleveref,aliascnt}
\usepackage{appendix}
\usepackage{verbatim}
\usepackage{mfirstuc}
\usepackage{xspace}
\usepackage{multicol}
\usepackage{tcolorbox}

\usepackage[ruled]{algorithm}
\usepackage{algorithmicx}
\usepackage{algpseudocode}
\algnewcommand{\IIf}[1]{\State\algorithmicif\ #1\ \algorithmicthen}
\algnewcommand{\EndIIf}{\unskip\ \algorithmicend\ \algorithmicif}

\algrenewcommand\algorithmiccomment[1]{\hfill$\triangleright$~#1}

\newcommand{\AllTasks}{[m]}

\newfloat{TCI}{htb}{lop}
\floatname{TCI}{Task-Completion Instance}

\makeatletter
\newenvironment{TaskInstance}[1][htb]{%
    \renewcommand{\ALG@name}{Task-Completion Instance}
   \begin{tcolorbox}[colback=black!4!white,colframe=black!4!white,top=-5pt]%
   \begin{TCI}[#1]%
  }{\end{TCI}\end{tcolorbox}}
\makeatother
\Crefname{TCI}{Task-Completion Instance}{Task-Completion Instances}
\crefname{TCI}{Task-Completion instance}{Task-Completion instances}

\theoremstyle{plain}
\newtheorem{theorem}{Theorem}[section]
\newtheorem{claim}[theorem]{Claim}

\newtheorem{lemma}[theorem]{Lemma}
\newtheorem{observation}[theorem]{Observation}

\newtheorem{definition}{Definition}

\newcommand{\F}{\mathbb{F}}
\newcommand{\Fp}{\mathbb{F}_p}

\DeclareMathOperator{\Hamm}{\mathrm{Hamm}}
\DeclareMathOperator{\poly}{poly}

\newcommand{\E}{\mathbb{E}}
\newcommand{\bad}{\ensuremath{\mathsf{bad}}}
\newcommand{\load}{\ensuremath{\mathsf{load}}}

\newcommand{\CF}{load balancing covering family\xspace}

\newcommand{\leadersafe}{executor-safe\xspace}  

\newcommand{\inTask}{\text{Inner Task}} 
\newcommand{\outTask}{\text{Outer Task}}
\newcommand{\Store}{\textsc{NetStore}}
\newcommand{\Retrieve}{\textsc{NetRetrieve}}

\newcommand{\ComputeOutgoingMessages}{\text{Compute-Messages}}

\newcommand{\Enc}{\mathrm{Enc}}
\newcommand{\Dec}{\mathrm{Dec}}

\newcommand{\ALG}{\mathcal{A}}

\begin{document}
\title{\vspace{-4ex}\textbf{The Task Completion Problem \\ and its Application to Crash-Resilient Computation}} 
\author{Orr Fischer \\
	\small Bar-Ilan University \\
	\small orr.fischer@biu.ac.il \\
	\and
	Ran Gelles \\
	\small Bar-Ilan University \\
	\small ran.gelles@biu.ac.il \\}
\date{}
\maketitle

\begin{abstract}
We study the \emph{Task Completion} problem, in which $M$ abstract tasks must be completed by a network of $n$ crash-prone nodes, where up to $\alpha n$ nodes may crash for some constant $\alpha<1$. 
Our main result is a deterministic congested-clique algorithm that completes all $M$ tasks in $O(\lceil M/n\rceil \log n)$ rounds. 
This round complexity is optimal up to $\log\log n$ terms.

The key technical ingredient underlying our algorithm is a novel combinatorial structure, which we call a \emph{load balancing covering family}.
In essence, this covering family induces, for each task, a subset of nodes responsible for attempting to complete it.
The properties of the load balancing covering family guarantee that, regardless of which tasks remain incomplete and which nodes crash, 
(i) no node is overloaded with incomplete tasks, and 
(ii) no task is left with too few potential assigned nodes.
This yields a balanced per-node workload and prevents non-crashed nodes from being concentrated on a small subset of tasks, thereby ensuring sufficient progress in completing the remaining tasks.

As an application of our task completion method,
we give a deterministic algorithm for simulating any $T$-round congested-clique algorithm  in the presence of up to $\alpha n$ crash faults in
$O(T^2 \log n + T \log^2 n)$ rounds.
This improves upon a recent result by Censor-Hillel et al.\ (DISC~2025), which requires
$T^2\cdot 2^{O(\sqrt{\log n}\log\log n)}$ rounds.
\end{abstract}

\thispagestyle{empty}
\clearpage
\setcounter{page}{1}
\section{Introduction}
\label{sec:intro}

One of the main advantages of distributed computation is the ability to parallelize work across multiple concurrent devices that operate simultaneously and independently, resulting in a significantly reduced overall completion time. In particular, if $n$ parallel tasks can be  distributed among a network of $n$ computing devices (nodes), and each task requires constant time, then completing all tasks in parallel reduces the total time to~$O(1)$.
This idealized scenario is often disrupted by \emph{faults}. Some nodes may malfunction and stop operating during the computation (i.e., nodes may \emph{crash}). Another potential challenge is \emph{congestion}: when communication links have bounded bandwidth, the remaining nodes may become bottlenecks for both incoming and outgoing communication.
In this work, we study crash-prone \emph{congested cliques}~\cite{LPPP05}, where up to an $\alpha < 1$ fraction of the $n$ fully connected nodes may crash, and seek efficient (ideally, optimal) algorithms for completing all tasks under the $O(\log n)$ per-message bandwidth restriction.

At first glance, and ignoring congestion for the moment, crashes might appear to have only a limited impact on the overall computation time, even when the fraction of failing nodes is large.
If the $\alpha n$ devices that may fail were known in advance, the remaining $(1-\alpha)n$ nodes could redistribute the workload among themselves and complete the $n$ tasks in $O(1/(1-\alpha))$ rounds. While this quantity grows as $\alpha \to 1$, it remains constant.

Unfortunately, this intuition is misleading, as it relies on advance knowledge of which $\alpha n$ nodes will crash, allowing the remaining nodes to take over their work. In practice, the set of crashed nodes is not known ahead of time and can occur adaptively throughout the computation. 
Moreover, nodes may fail in the middle of a round, leading to inconsistent views among the remaining nodes. As a result, at any given round, there may be no consensus on which nodes have crashed.

When up to $\alpha n$ nodes may crash arbitrarily, the remaining nodes must coordinate in order to detect these failures and redistribute the yet-incomplete tasks. This coordination itself incurs a nontrivial cost: it requires $\Omega(\log n)$ rounds, even if all nodes agree in every round on which nodes have crashed so far and, consequently, on which tasks remain incomplete; see Section~3.2 of~\cite{GS08}. 

\paragraph{Message-Passing Task Completion and Prior Work.}
The \emph{Task Completion Problem} (also known as \emph{Do-All}) considers the setting in which $M$ \emph{tasks} must be completed by a fully-connected network of $n$ nodes despite crash failures.\footnote{To simplify the presentation, in this section we set $M = n$. We stress that all results also apply to the general case.}
Tasks are \emph{abstract} entities that can represent actions within the model (e.g., sending a message or updating a local variable), but may also be used to represent activities outside the model, such as accessing a shared resource, writing to an external device, etc. 
A task is typically specified by a predicate whose satisfaction indicates that the task is complete, together with pseudocode describing how the completion of the task can be achieved. 
Examples of such tasks appear in Task-Completion Instances~\ref{task:ComputeOutgoingMessages}--\ref{task:outTask} in \Cref{sec:simulation}.
Tasks need to be \emph{independent}: any node can complete any task, with no prescribed order or dependencies. Moreover, up to $n$ tasks may be completed simultaneously (e.g., when all nodes attempt to complete different tasks in the same round).
Tasks need also to be \emph{idempotent}: if a task is completed multiple times or concurrently by multiple nodes, the outcome is the same and this event is counted as a single  correct completion of the task. 
We assume completion of a single task takes $R\ge1$ rounds and require tasks to be
\emph{\leadersafe}: 
if a node attempts to complete a task, and that node does not crash during the $R$ rounds of its execution, then the  completion of the task is guaranteed regardless of failures of other nodes. Note that the task may involve actions of other nodes, however, the task is \leadersafe if their failure does not prevent the correct completion of the task.

The task completion problem was introduced for synchronous message-passing systems by Dwork, Halpern, and Waarts~\cite{DHW92}.
They designed three different algorithms that tradeoff time and message complexity while aiming to optimize the \emph{work}, defined as the total number of attempts to complete tasks (each node can perform one unit of work per round or remain idle).
Assuming $R=1$, their first algorithm runs in $O(n)$ rounds,   
while the second requires exponentially many rounds. 
The round complexity of their third algorithm depends on the number of faults that actually occur: it takes $O(n^2)$ rounds in the worst case, but only $O(1)$ rounds in the fault-free scenario.
Later works~\cite{DMY94,CDS97,DDGHS17,CGKS17} typically do not analyze the total round complexity, but instead state only  the work, which implies loose bounds on the round complexity. 
These works take at least $\Omega(\log^2 n)$ rounds to complete all tasks.\footnote{Note that many of these works tolerate a number of failures $f < n$, which trivially implies $\Omega(n)$ rounds for the case where only a single node remains to complete all tasks. For a fair comparison, we consider bounds that explicitly depend on the number of failures $f$ and substitute $f = \alpha n$.}
The work of Georgiou, Russell, and Shvartsman~\cite{GRS04} suggests that $O(\log n)$ rounds may suffice if the nodes are equipped with an oracle for incomplete tasks, that is, if the nodes are always in consensus regarding which tasks remain incomplete. 
However, finding an algorithm with $O(\log n)$ round complexity without access to such an oracle remains an important open problem, and to the best of our knowledge, no deterministic $O(\log n)$-round solution in the message-passing model is currently known.


We answer the above open question affirmatively by presenting a deterministic message-passing algorithm for the task completion problem that completes in $O(\log n)$ rounds when $R=1$ in the presence of up to $\alpha n$~crashes without any other assumptions. 
Furthermore, in our algorithm, nodes communicate using messages of size $O(\log n)$ bits; that is, the algorithm fits within the \emph{congested clique} model~\cite{LPPP05}. We can now, somewhat informally, state our main result as follows.
\begin{theorem}
\label{thm:main-informal}
For any $\alpha < 1$, if completing a single task takes $R\ge1$ rounds, the $n$-task completion problem can be solved by a deterministic congested clique algorithm in $O(R\log n)$ rounds in the presence of up to $\alpha n$ crash-faults.
\end{theorem}

We note that this result is tight up to $\log\log n$ factors. Indeed, a lower bound of $\Omega(\log n / \log\log n)$ rounds for the $n$-task completion problem with $R=1$ was established by Georgiou, Russell, and Shvartsman~\cite{GRS04,GS08}, even in the stronger setting where nodes share a global view of the incomplete tasks via an oracle. In \Cref{app:lowerbound}, we revisit this proof in a slightly simplified form for the sake of completeness.

While the task completion problem is a fundamental distributed problem in its own right, it is also closely connected to resilient simulation of algorithms, particularly in the PRAM model (see, e.g.,~\cite{GS08,KS13} and the related work section below). We leverage our congested-clique task completion algorithm to obtain a method for simulating any congested-clique algorithm in crash-prone networks, while maintaining low bandwidth usage to avoid congestion. This application is summarized in the following theorem.
\begin{theorem}
\label{thm:main-application-informal}
For any $\alpha < 1$, any congested-clique algorithm that runs in $T$ rounds can be simulated in the presence of up to $\alpha n$ crash-faults in $O(T^2\log n + T \log^2 n)$ rounds. 
\end{theorem}
The above theorem improves upon a recent result by Censor-Hillel, Fischer, Gelles, and Soto~\cite{CFGS25}, which simulates any $T$-round algorithm in $T^2 \cdot 2^{O(\sqrt{\log n}\log\log n)}$ rounds. 
In another recent work, Censor-Hillel and Soto~\cite{CS25} show that
in certain special cases, structural properties of the simulated algorithm allow for further optimization, yielding simulations with only polylogarithmic (in~$n$) overhead. 
In the same spirit, our simulation technique also achieves improved performance in specific cases. 
In particular, when the state maintained by each node has size at most $O(n \log n)$ bits, our simulation runs in $O(T \log^2 n)$ rounds, incurring only polylogarithmic overhead.

\subsection{Our Techniques}
\paragraph{A {\CF}.}
\label{sec:intro-cf}
Our task completion algorithm relies on a new combinatorial structure that we call a \emph{\CF}.
This covering family consists of sets $A_1, \ldots, A_m$, where each $A_i$ is a subset of $\{1, \ldots, n\}$. The covering family is parameterized by three values: integers $k$ and $B$, and a fraction $\epsilon \in (0,1)$, and it satisfies two important and useful properties (see Definition~\ref{def:covering_family} for a formal statement).
(i) Each set $A_i$ is neither too large nor too small; specifically, for all $i$, $\frac{nB}{2k} \leq |A_i| \leq \frac{2nB}{k}$.
(ii) Any $k$ sets $\{A_{i_1}, \ldots, A_{i_k}\}$ are load-balanced with respect to a set $J \subseteq \{1, \ldots, n\}$ of size at least $(1-\epsilon)n$. By ``load-balanced'' we mean that each element $j \in J$ appears in approximately $B$ of the subsets $\{A_{i_1}, \ldots, A_{i_k}\}$. Formally, for all $j\in J$,
\[
(1-\epsilon)B \leq \left|\{\ell \in \{1,\ldots,k\} \mid j \in A_{i_\ell}\}\right| \leq (1+\epsilon)B\text{.}
\]

\paragraph{A \CF--based task completion.}
We use the \CF as follows. Assume there are $n$ nodes and $m$ tasks in total, out of which $k\le m$ tasks are incomplete. For the time being, suppose that the nodes have full knowledge of which tasks still need to be completed; we later show how to remove this assumption. 
The nodes construct a $(k,B,\epsilon)$-\CF with $0 < \epsilon < 1-\alpha$ and $B = O(1)$. The existence of such a family is established in Section~\ref{sec:coverfamily} using the probabilistic method.

Given this construction, task completion proceeds as follows. For $2B\cdot R$ rounds, node $j$ attempts to complete all incomplete tasks $i \in \{1, \ldots, m\}$ such that $j \in A_i$. That is, for each task $i$, the set $A_i$ specifies the nodes assigned to complete task~$i$. 
The second property of the \CF guarantees that all but $\epsilon n$ nodes are assigned at most $(1+\epsilon)B$ incomplete tasks and can therefore complete all of them within these $2B\cdot R$ rounds, unless they crash (recall, completing a single task takes $R\ge 1$ rounds). 
Moreover, since $\epsilon < 1-\alpha$, even if $\alpha n$ nodes have already crashed, the set $J$ still contains $(1- \epsilon-\alpha)n = \Omega(n)$ nodes that can make progress during this segment. 
The first property of the \CF ensures that nodes are well distributed across tasks: no task is assigned too many nodes, and the assignments are instead evenly spread. Hence, the $\Omega(n)$ non-crashed nodes in~$J$ are not all assigned to the same task(s) and can make substantial progress during this step. In particular, we prove that at least $\epsilon k$ of the incomplete tasks are completed.

After this step, we recurse on the remaining incomplete tasks. Each step reduces the number of incomplete tasks by a factor of~$\epsilon$, and therefore after $O(\log_{1/\epsilon} n)$ steps, all tasks are guaranteed to be completed. The algorithm thus completes in $O(2B\cdot R\cdot  \log_{1/\epsilon} n)$ rounds, and since $B = O(1)$, we obtain the stated $O(R\log n)$-round algorithm.

\paragraph{Inconsistent views and optimistic computation.}
The above intuitive algorithm relies on the nodes knowing, in every step, which $k$ tasks remain incomplete. However, in our model, such perfect knowledge is unavailable. Although the network is fully connected and nodes may communicate in every round, this only provides partial information about failures. In particular, if a node crashes in round~$r$, some nodes may learn of the crash at the end of round~$r$, while others learn of it only in round~$r+1$.
More importantly, discrepancies in the nodes’ views of task completion can be substantial. For example, suppose nodes report to the entire network which tasks they have completed and then crash during this reporting step. In this case, only a subset of the nodes receive these reports, leading to disagreement across the network regarding which tasks remain incomplete. We note that approaches that attempt to redistribute tasks among non-crashed nodes (e.g.,~\cite{GRS04,GS08,CFGS25}) typically rely on all nodes sharing a consistent global view.

One possible solution is to run a crash-resilient consensus algorithm after every step to synchronize views. However, this approach is prohibitively expensive in terms of round complexity, adding $\Omega(n)$ rounds in the worst case. Instead, we take a different approach.

We execute the \CF--based approach described above without knowing the exact number~$k$ of incomplete tasks and without requiring full agreement on which tasks have been completed. 
To address the issue of inconsistent views, we first redefine the notion of task completion. A task is said to be \emph{fully verified} if at least one node has completed it and \emph{all} non-crashed nodes have been informed that it is complete. 
With this re-definition, the set of tasks that a given node knows to be complete may be a strict superset of the set of fully verified tasks, since some completion notifications may not have reached all nodes.
%
Nevertheless, this suffices for our purposes. 
The properties of the \CF guarantee that each node is assigned approximately $B$ tasks among the $k$ tasks it believes to be incomplete (which, as noted above, form a superset of the not-yet fully verified tasks). In particular, since property~(2) of the \CF holds for \emph{any} choice of $k$ subsets, this approach remains valid even when different nodes maintain different views of which $k$ tasks are incomplete.
After completing its assigned tasks, the node informs all other nodes that \emph{all tasks} assigned to it have been completed. 
Thus,
unless the node crashes during this process, this report makes all of these tasks fully verified regardless of the node's prior knowledge  about their status.

It remains to explain how the algorithm operates without consistent knowledge of~$k$.
Our algorithm proceeds by \emph{guessing} the number $k$ of not-yet fully verified tasks up to a factor of $\epsilon$ and optimistically attempting to fully verify them. 
In each step, the guessed value decreases by a factor of $\epsilon$, until it reaches a constant after $O(\log n)$ steps. For any step in which the guess is sufficiently accurate (i.e.,  $\epsilon$-close to the real value), the network makes substantial progress by fully verifying at least $\epsilon k$~tasks. This guarantees that the true number of not-yet fully verified tasks is always upper bounded by the algorithm’s current guess. Moreover, whenever the true number approaches this upper bound, it necessarily lies within the range that ensures progress.

\subsection{Application: Robust Congested-Clique Simulation}
\label{sec:introApplications}
In recent work, Censor-Hillel, et~al.~\cite{CFGS25} presented a general method for simulating any congested-clique computation on a network of size $n$ in the presence of up to $\alpha n$ crash faults. Given an algorithm~$\ALG$ with round complexity~$T$, their approach converts $\ALG$ into an equivalent logic circuit $C$, which is then simulated by the network in a gate-by-gate manner.
Specifically, the gates in the lowest uncomputed layer of $C$ are distributed among the nodes, which compute them and \emph{store} the resulting output in the network using a locally decodable error-correction code (LDC). 
Note that, in order to distribute the gates in a consistent manner, all nodes must agree on the set of non-crashed nodes. To this end, the simulation in~\cite{CFGS25} assumes \emph{clean crashes}, in which nodes may crash only at the beginning of a round (before sending any messages), but never in the middle of a round.

To handle crashes, their simulation employs a doubling strategy: 
whenever a node crashes, \emph{two} other nodes take over the gates assigned to the crashed node.
Moreover, a node may fail to compute a gate if some of the nodes holding (coded parts of) its input values have crashed. In this case as well, the doubling strategy is applied, with the node doubling its attempts to compute the gate by contacting additional subsets of nodes that store the same information.
The analysis in~\cite{CFGS25} shows that each doubling loop succeeds after $O(\log n)$ iterations for any crash pattern with at most $\alpha n$ crashes, for any~$\alpha < 1$. 
As a result, the overall simulation requires $O(T^2\cdot \log^2 n \cdot q)$ rounds: the circuit $C$ has $T^2$ layers of gates, each layer incurs an $O(\log^2 n)$ overhead due to the doubling strategy, and $q = 2^{O(\sqrt{\log n}\log\log n)}$ denotes the round complexity required to decode information stored using the LDC without causing congestion.

An immediate application of our task completion algorithm to the approach of~\cite{CFGS25} is to view each gate as a task. On the surface of it, sequentially applying the task completion algorithm $T^2$ times then yields a simulation with the same asymptotic round complexity of~\cite{CFGS25}, but without relying on the clean-crash assumption made in~\cite{CFGS25}.

We can do even better. 
Rather than converting $\ALG$ into a circuit and simulating it gate by gate, we adopt a more direct approach that simulates $\ALG$ round by round using our algorithm for the task completion problem.  
Consider a single round of~$\ALG$, in which the $n$ nodes each send one message to each of their neighbors. We view this round as consisting of $n$~tasks, where each node has the task of obtaining all messages designated to it during that round. Once all these tasks are completed, we can conclude that the round of~$\ALG$ has been correctly simulated.

However, this translation into tasks does not directly satisfy the assumptions required by our task-completion algorithm. In particular, in this setting a node may fail to obtain all of its designated messages (and thus, fail to complete its task) even if it \emph{does not} crash itself. That is, the task is not \emph{\leadersafe}. 
Such a failure may occur when some of the nodes responsible for sending these messages crash before delivering them. 
By contrast, in the standard task completion problem, the only reason a task may fail to be completed is that the node assigned to execute it crashes. This mismatch requires additional care in adapting the task completion algorithm to the round-by-round simulation setting.

Our solution addresses this issue by decomposing the required communication into two \emph{nested} subtasks, each of which is carried out using our task-completion algorithm. The first subtask, which we refer to as the \emph{inner} task, assumes that all messages sent by node~$i$ in the $r$-th simulated round, $M_i(r)$, are already stored in the network using a standard error-correction code (ECC) of block length~$n$. 
In this encoding, each node holds a single symbol of the codeword~$C_i(r)=\Enc(M_i(r))$, 
and the original messages $M_i(r)$ can be recovered even if up to $\alpha n$ symbols of~$C_i(r)$ are missing; see Section~\ref{sec:ECC}. 
The $i$-th inner task consists of decoding the codeword~$C_i(r)$ to obtain all the $n$ messages sent by node $i$ in the $r$-th simulated round of~$\ALG$, and then delivering to each node~$j$ the message designated to it by node~$i$.

This inner task alone is insufficient if node~$j$ has already crashed: in that case, it is unclear where the message designated to~$j$ should be delivered. To resolve this issue, we introduce a second subtask we call the \emph{outer} task. The $j$-th outer task is to \emph{act as node~$j$} in the simulation of~$\ALG$, that is, to collect all messages designated to~$j$ and generate the messages that $j$ needs to send in the next round,~$M_j(r+1)$. 
Specifically, once a node decides to act as the node~$j$ in~$\ALG$ in the outer task, all nodes initiate, in a nested manner, a task-completion instance for the inner task. 
The inner task guarantees that the node acting as~$j$ receives all $n$ messages intended for~$j$ in round~$r$ of~$\ALG$, which we denote by~$S_j(r)$.
The node then stores $S_j(r)$ in the network  using an ECC (for technical reasons explained later). 
Finally, the node computes all messages that~$j$ should send in the next round of~$\ALG$, $M_j(r+1)$, encodes these messages using an ECC, and stores the resulting codeword in the network by sending one symbol to each node. 
This codeword subsequently serves as the input for simulating the next round of~$\ALG$.

We note that both the inner and outer tasks are idempotent. Specifically, if multiple nodes complete the same task and send messages or store codewords in the network, they transmit or store identical information. As a result, no inconsistency is introduced by redundant executions of the same task.
Moreover, both tasks are \leadersafe. 
In particular, although the node simulating node~$j$ in the outer task relies on other nodes to deliver the required messages, it does not depend on any specific helper node. 
Instead, message delivery is carried out through the task-completion mechanism of the inner task in a fault-safe manner. 
Consequently, if any node responsible for delivering messages crashes, it is automatically replaced by other nodes via task completion. This ensures that the outer task completes successfully as long as the node simulating node~$j$ itself does not crash.

As described above, the outer and inner tasks are executed in a nested manner. First, each node participates in the outer task algorithm and obtains an identity $j$ to simulate. The node then informs all other nodes that it is simulating identity~$j$ (multiple such nodes may exist). Once all nodes have announced the identities they simulate, the network initiates an instance of the inner task, whose goal is to deliver messages to the simulated identities. This phase takes $O(\log n)$ rounds using our task completion algorithm for the inner task, and by its conclusion, all simulated identities hold all messages required for that round.
The outer task then proceeds to its next step, during which non-crashed nodes assume identities that have not yet been simulated. After $O(\log n)$ steps of the outer task, all identities will have been simulated. Consequently, simulating a single round of~$\ALG$ requires $O(\log^2 n)$ rounds.

Unfortunately, this does not yet imply an $O(T \log^2 n)$-round simulation of~$\ALG$. The difficulty arises from the fact that a node simulating identity~$j$ cannot generate the outgoing messages of~$j$ in round~$r$ of~$\ALG$ solely from the messages that~$j$ received in round~$r-1$. Rather, to compute these messages, the simulating node must have access to the complete \emph{state}  of node~$j$ at the beginning of round~$r$.
There are several ways to address this issue. If the state (i.e., memory) of node~$j$ is short, e.g., of size $O(n \log n)$, then storing this state at each round together with the codeword~$C_j(r)$ is a viable option, yielding an overall $O(T \log^2 n)$-round simulation of~$\ALG$. However, in the general case, storing and decoding the full state is prohibitively expensive.
A more general solution is for the node simulating identity~$j$ to first reconstruct the state that~$j$ would have at the beginning of round~$r$. This can be achieved by obtaining the information $S_j(1),...,S_j(r-1)$, which correspond to the messages that~$j$ received in all rounds prior to~$r$. As a result, simulating round~$r$ of~$\ALG$ requires $O(r + \log n)$ rounds to complete the outer task.
Overall, the obtained round complexity is~$O(T^2 \log n + T\log^2 n)$.

\subsection{Other Related Work}
\label{sec:relatedwork}
Several extensions of the task completion problem have appeared in the literature.
Randomized algorithms were introduced in~\cite{CK04,CGKS08}, where randomization reduces the total work to $n \log^2 n$ against an arbitrary adaptive adversary with $f < n$ failures, and to $O(n \log^* n)$ against a weakly adaptive adversary that selects $f$ nodes prior to the start of execution and may crash only these selected nodes.
Variants of task completion with crashing nodes that can later restart were studied in~\cite{CDS97}.
Algorithms tolerating Byzantine faults rather than crash faults were presented in~\cite{FGRS05}.

A variant of the task completion problem in the PRAM shared-memory model is known as the \emph{write-all} problem~\cite{KS89}. In this setting, $n$ memory cells are initially set to $0$, and $n$ processors are required to write the value $1$ to all cells~\cite{AW97,KPRS91,BKRS96}. 
A common technique in this line of work is to assign each processor a predetermined permutation of $\{1,\ldots,n\}$, which specifies the order in which it attempts to write to the memory cells. A processor skips any cell that it already knows has been written. Optimizing such permutations~\cite{NR95,KMS05} amounts to efficiently finding a collection of permutations that minimizes the total number of write attempts.
In this sense, our combinatorial \CF can be viewed as an extension of this approach: rather than fixing a single permutation, it effectively induces randomized task orders with desirable properties, such as sufficient redundancy and low congestion. We note that the above works typically assume a model in which up to $n-1$ processors may fail. In this extreme setting, the last remaining processor must eventually perform all tasks, which motivates the use of permutations over the entire task set.
Simulations of arbitrary PRAM algorithms using write-all primitives were presented in~\cite{Shvartsman89,KPS90,KS13}. 

The congested clique model~\cite{LPPP05} has been extensively studied in the literature. A large body of work explores a wide range of algorithmic questions and graph-theoretic problems, including minimum spanning tree computation~\cite{LPPP05,N21,HegemanPPSS15,GhaffariP16,Korhonen16,JN18}, routing~\cite{Lenzen13}, graph coloring~\cite{Parter18,ParterS18,CFGUZ19,BKM20,CzumajDP21,CoyCDM23}, 
subgraph detection and enumeration~\cite{DLP12,IzumiG17,Pandurangan0S18,FischerGKO18,Censor-HillelFG20,Censor-HillelGL20,Censor-HillelFG22}, and many others.
Despite this extensive body of work, significantly less is known about fault tolerance in the congested clique model. Only a limited number of recent works address resilience to failures in this setting~\cite{AMPV22,KMS22,KM23,MR23,FP25,CFGS25,CS25}.

\subsection{Paper Outline}
After presenting the model and preliminaries in \Cref{sec:prelim}, we define in Section~\ref{sec:coverfamily} the notion of a \emph{\CF} and prove its existence for suitable parameters.  
In Section~\ref{sec:taskcompletion} we present a deterministic task completion algorithm based on these coverings, prove its correctness, and analyze its round complexity. A lower bound on the round complexity of any deterministic task-completion algorithm is given in \Cref{app:lowerbound}.  
Finally, in Section~\ref{sec:simulation} we apply the task completion framework to obtain a crash-resilient simulation of arbitrary $T$-round congested-clique algorithms, with a total complexity of $O(T^2\log n + T\log^2 n)$ rounds.

\section{Preliminaries}
\label{sec:prelim}
For an integer $n\ge1$, we let $[n]$ denote the set $\{1,2,\ldots, n\}$.
All logarithms are taken to base~2 unless otherwise mentioned. A string $S$ over an alphabet $\Sigma$ is a consecutive sequence of 0 or more letters $S=s_1,s_2,\ldots$ where each letter $s_i\in \Sigma$. The length of the string is denoted $|S|$. 
For $1\le i \le |S|$, 
we let $S[i]=s_i$ denote the $i$-th letter in~$S$. 

\subsection{Network, Crash-faults, and the Congested Clique Model}
\label{sec:prelim-Network}
We assume a fully connected network $G=(V,E)$ of size $|V|=n$, where all nodes know the identities of all their neighbors. 
For the ease of notation we identify the set of nodes $V$ with $[n]$.
Communication follows the standard synchronous CONGEST model, namely, 
communication occurs in synchronous rounds, where in each round, each node can send $O(\log n)$ bits to each one of its neighbors.

A node may crash at any time of the execution. If a node $v$ crashes during some round,
then an arbitrary subset of its outgoing messages may fail to be sent in that round and $v$ cannot send messages in any future round.
If a message fails or is not sent at all,  the recipient node obtains the symbol~$\bot$. Note that if the recipient was expecting a message from~$v$, the reception of  a~$\bot$ indicates that the sender~$v$ has crashed.  
We assume that at most $\alpha n$ nodes may crash over the entire execution of the protocol, for a fixed and known constant parameter $\alpha\in [0,1)$.

As noted in prior work~\cite{CFGS25,CS25}, special care is required to prevent the loss of nodes’ inputs if they crash before the computation begins. Existing approaches address this issue either by guaranteeing that no node crashes during the first few rounds~\cite{CS25}, or by modifying the model so that inputs are provided in encoded form and stored by the network~\cite{CFGS25}. We adopt the latter approach, as it yields better composability.
In our model, the input of each node has length $O(n \log n)$ bits and is encoded using an error-correction code (see \Cref{sec:store_retrieve} and \Cref{sec:ECC} below). Each node holds a single symbol of each encoded input. 
Consequently, even if a node crashes before the computation starts, its input can still be recovered by querying the corresponding codeword symbols.
Similarly, we require that the output of each node be encoded and stored in the same manner. Then, composing two algorithms, where the output of the former serves as the input to the latter, is straightforward.

\subsection{The Task Completion Problem}
\label{sec:task_completion_def}
In the task completion problem, we have $M$ abstract tasks with unique labels~$[M]$. Initially, all tasks are defined as \emph{incomplete}. 
At the start of any round, each node $v$ may decide to complete \emph{any} 
task of its choice. 
The completion of a task takes $R\ge1$ rounds, after which the task is called  \emph{completed}.
While tasks represent arbitrary activities (even outside the prescribed model), 
we only consider  tasks that satisfy the following three properties. Tasks must be \emph{independent}: the completion of any task does not affect any other task, and any task can be executed concurrently with any other task; further, up to $n$ tasks can be executed concurrently. The tasks must also be \emph{idempotent}: each task can be executed one or more times by one or more nodes and will produce the same final result; in particular, completing an already completed task does not change its status or compromise its validity.
Finally, tasks must be \emph{\leadersafe}:  a task  completes successfully whenever the node executing it does not crash, regardless of failures of other nodes.

We say that a deterministic algorithm solves the task completion problem if, once it finishes running, every task has been completed, regardless of how crashes occur during its execution, provided that there are at most $\alpha n$ such crashes.

\subsection{Error Correction Codes}
\label{sec:ECC}

For an alphabet~$\Sigma$, the Hamming distance of two strings $x,y\in (\Sigma \cup \bot)^*$ of the same length, i.e., $|x|=|y|$,  is the number of indices for which $x$ and $y$ differ and is denoted by $\Hamm(x,y)={|\{ i \mid x[i] \ne y[i]\}|}$. For two strings $x \in \Sigma^*$, $y \in (\Sigma \cup \{\bot\})^*$ and a value $c\in (0,1)$, we say that $y$ can be obtained by a $c$-fraction of erasures from $x$ if $|x| = |y|$, and for all $i \in [|x|]$, it holds that either $x[i] = y[i]$ or $y[i] = \bot$, where the latter case happens at most $c|x|$ times. An index $i$ in which $y[i] = \bot$ is called an erasure.
    
    For a prime power $p$, we denote by~$\Fp$ the finite field of size~$p$. An \emph{error correction code} is a mapping $\Enc: \Fp^K \to \Fp^N$ that takes $K$ symbols of the alphabet $\Fp$ into $N$ symbols of the alphabet~$\Fp$.
    The value $N$ is called the  block length of the code. The ratio $\rho = K/N$ is called the \emph{rate} of the code. The \emph{relative distance} of a code is the minimal normalized Hamming distance between any two codewords, denoted $\delta=\min_{ m\ne m'} \tfrac1N {\Hamm(\Enc(m),\Enc(m'))}$. 
    The following claim is well known,
\begin{claim}[MDS Erasure Correction Codes~\cite{RS60}]
    \label{fct:ErasureECC}
    For any $N>K>1$, there exists a code
    $\Fp^K\to\Fp^N$, where $p\ge N$, with relative distance $\delta = (N-K+1)/N$. Such a code can correct up to $N-K$ erasures.
\end{claim}
As a corollary, for any $\alpha<1$ and a sufficiently large~$N$, by choosing the right parameters, we can have a code with block length~$n$ and relative distance $\delta>\alpha$, that can correct $\alpha$-fraction of erasures. 
Further, each symbol in such a code is of length $O(\log N)$ bits.

\section{Load Balancing Covering Sets}
\label{sec:coverfamily}
We begin by defining a $(k,B,\epsilon)$-\CF and then prove its existence for appropriate choices of the parameters.
\begin{definition}
\label{def:covering_family}
    Let $1\le m \le n$ be two fixed integers. 
    For any integer parameters $B,k \geq 1$, and value $\epsilon \in (0,1)$, 
    we say that a family $A_1,\dots,A_m \subseteq [n]$ is $(k,B,\epsilon)$-\CF if the following conditions are satisfied:
    \begin{enumerate}
        \item \label{item:covering_family_1} $\frac{nB}{2k} \leq |A_i|  \leq \frac{2nB}{k}$.
        \item \label{item:covering_family_2}  
        For any $k$ sets $A_{i_1},\dots,A_{i_k}$ of this family, there exist a set $J \subseteq [n]$ of size $|J| \geq (1-\epsilon) n$, such that for all $j \in J$, 
    \end{enumerate}
    \[(1-\epsilon)B \leq |\{\ell \in [k] \mid j \in A_{i_\ell} \}| \leq (1+\epsilon)B.\]
\end{definition}

\begin{lemma}
\label{lem:uniform_covering}
    Let  $\epsilon \in (0,1)$ be a constant. Let  $c = \epsilon^3/8$, and $B = \lceil 72/c \rceil$. 
    For any sufficiently large~$n$, any integer $m \leq cn$ and any $B \leq k \leq m$, there exists a $(k,B,\epsilon)$-\CF. 
\end{lemma}
\begin{proof}
    We prove the existence of such a family using the probabilistic method. 
    For each set $A\in \{A_1,\dots,A_{2m}\}$ we construct $A$ by including in it each element $j\in[n]$  with probability 
    $B/k$, independently across elements and sets. We show that with good probability, most of these sets are of size $|A_i| \in \left [\frac{nB}{2k},\frac{2nB}{k}\right]$, and that 
    \Cref{def:covering_family}(\ref{item:covering_family_2}) holds for any choice of $k$ sets from these $2m$ sets. We then remove $m$ sets so that all remaining sets satisfy \Cref{def:covering_family}(\ref{item:covering_family_1}), while the property in \Cref{def:covering_family}(\ref{item:covering_family_2}) endures. 

    First, we prove that \Cref{def:covering_family}(\ref{item:covering_family_1}) holds for more than half the sets, with probability strictly larger than~$1/2$. 
    By linearity of expectations, for any $i\in[m]$ it holds that $\E[|A_i|]= nB/k$. Then, it follows that
    
    \begin{equation}
    \label{eq:temp_single_set_good_size}
        \Pr\left(\frac{nB}{2 k} \leq |A_i|  \leq \frac{2nB}{k}\right) \geq 1-2e^{-nB/(12k)} \geq 0.99\text{,}
    \end{equation}
    where the first inequality follows by Chernoff's inequality (\Cref{thm:chernoff}(\ref{item:chernoff_two_sided}), with $\delta = 1/2$), and the next inequality by recalling that $k\le m<n$, and by taking~$B$ to be sufficiently large, e.g., $B > 72/c$. Let $M$ be the number of sets that are of size $\in [\frac{nB}{2 k},\frac{2nB}{k}]$. We note that by \Cref{eq:temp_single_set_good_size}, we have $\E[M] \geq 0.99 \cdot 2m$. Thus, the probability that $M \geq m$ is at least
    \[\Pr(M \geq m) = 1- \Pr(M < m) \ge 1- \Pr(M < (1-0.98)\E[M]) \ge 1- e^{-1.98m \cdot (0.98)^2/2}\text{,}
    \]
    where the last step  follows by Chernoff's inequality (\Cref{thm:chernoff}(\ref{chernoff:smaller_side}), with $\delta = 0.98$). For $m\ge1$, this is above~$0.6 > 1/2$.

    Next, we show that \Cref{def:covering_family}(\ref{item:covering_family_2}) holds with probability strictly larger than~$1/2$, for any $k$ sets out of all possible $2m$~sets. 
    Fix $k$ sets $A_{i_1},\dots,A_{i_k}$ from $\{A_1,\ldots, A_{2m}\}$ and let $I=\{i_1,\ldots, i_k\}$. For any $j\in[n]$ 
    we define the \emph{load} of $j$ in $I$ 
    to be
    \begin{equation}
    \label{eqn:load}
    \load(j,I) = |\{\ell \in I \mid j \in A_{\ell} \}|\text{,}
    \end{equation}
    namely, the number of subsets $\{A_i\}_{i\in I}$ that include the element~$j$. 
    Let $B(j,I)$ be the expectation of the load of $j$ in $I$,
    \[
     B(j,I)\coloneqq 
     \E[\load(j,I)]
     = k \cdot (B/k)=B\text{.}
    \]
    Note that this expectation value does not depend on the specific values of  $i_1,\ldots, i_k$ and~$j$ (but it does depend on $k=|I|$);  thus, $B(j,I)=B$ for all $j\in[n]$ and $I$'s of size~$k$.
    By Chernoff's inequality (\Cref{thm:chernoff}(\ref{item:chernoff_two_sided})),  for any fixed $j \in [n]$,
    \begin{equation}
    \label{eq:temp_chance_to_be_in_J}
        \Pr(|\load(j,I)-B| > \epsilon B) < 2e^{-\epsilon^2 B/3} < \epsilon/2.
    \end{equation}
    
    For $j$ and $I$ as above, set the indicator $\bad(j)$ to be the event where the index $j$ has a load which is $\epsilon$-far from its expected value, i.e., the event
    \begin{equation*}
        \load(j,I) \notin [(1-\epsilon)B,(1+\epsilon)B]\text{.}
    \end{equation*}
    Finally, set~$J$ to be all the  indices $j\in[n]$ which are not bad,  namely,  
    \(
    J \coloneqq [n] \setminus \{j \mid  \bad(j)
    \}\text{.}
    \)
    To complete the proof, we need to show that $|J| \ge (1-\epsilon)n$; let us then bound the probability that $|J| < (1-\epsilon)n$. 
    For a fixed $j \in [n]$, we have by \Cref{eq:temp_chance_to_be_in_J} that
    \[
    \Pr \left(j \in J\right) = \Pr(|\load(j,I)-B| \leq \epsilon B) \geq 1-\epsilon/2. 
    \]
    By linearity of expectation, 
    \(
    \E[|J|] = \sum_{j\in[n]}\Pr(j\in J)\geq (1-\epsilon/2)n\text{.}
    \)
    For $j\in[n]$, consider the events that $j \in J$  and note that these events are mutually independent for $j'\ne j$. 
    Then, by Chernoff's  inequality (\Cref{thm:chernoff}(\ref{chernoff:smaller_side})), 
    we get
    \[
    \Pr\left(|J| < (1-\epsilon)n\right) \leq \Pr\left(|J| < (1-\epsilon/2)^2n\right) \leq e^{-(\epsilon/2)^2(1-\epsilon/2)n/2}\text{.}
    \]

    By taking union bound over all the choices of $I\subset[2m]$ of size $|I|=k$, i.e., all possible  $k$ subsets of the family,
    we get that they all satisfy the condition 
    simultaneously with probability at least $1-e^{-(\epsilon/2)^2(1-\epsilon/2)n} \binom{2m}{k} > 1/2$,
    where the inequality follows by the fact that $m \leq cn$ with a sufficiently small  $c > 0$ with respect to~$\epsilon$. 
    Namely, 
    $\binom{2m}{k}  \leq 2^{2m} \leq 2^{2cn} < e^{(\epsilon/2)^2(1-\epsilon/2)n-1}$. 

    Summing it all up, 
    by a union bound, the probability that a randomly chosen $\{A_1,\ldots,A_{2m}\}$ does not satisfy \Cref{def:covering_family}(\ref{item:covering_family_2}) or has less than $m$ sets of size satisfying 
    \Cref{def:covering_family}(\ref{item:covering_family_1}), is \emph{strictly smaller} than  $1/2+1/2 = 1$. Hence, there exists a \CF with the desired parameters.
    Note that the proof holds for any sufficiently large $B=\Omega(1/\epsilon^3)$, that is bounded above by $B\le k \le m$.
\end{proof}

\section{A Deterministic Task-Completion Algorithm}
\label{sec:taskcompletion}
In this section we present our algorithm for the task completion problem, depicted in \Cref{alg:task-completion}, and prove the following.

\begin{theorem}
\label{thm:main_task_complete}
Assuming that completing a single task requires $R\ge1$~rounds, for any $\alpha \in [0,1)$ there exists a deterministic algorithm that solves the $M$-task completion problem in a congested clique of size~$n$ in the presence of up to $\alpha n$ crashes, using $O(R \lceil M/n \rceil \log n)$ rounds.
\end{theorem}

\subsection{Algorithm Description}
Let $\epsilon = \epsilon(\alpha) > 0$ be a sufficiently small constant relative to~$\alpha$, and let $c = c(\epsilon) > 0$ be a sufficiently small constant relative to~$\epsilon$. 
In this section, we assume that the number of tasks is $m = cn$. 
However, any number $M$ of tasks can be handled by partitioning them into sets of size at most $cn$ and completing each set sequentially (see \Cref{sec:completingTaskThm}).

\begin{algorithm}[htp]
\caption{Iterative Task Completion for $n$ nodes and $m=c n$ tasks (code for node $v$)  
}
\label{alg:task-completion}
\begin{algorithmic}[1]
\Statex \textbf{Inputs:} 
The maximal fault-tolerance  constant $\alpha>0$, the number of parties~$n$, the number of tasks $m =c n$.
Parameters $B$, $\epsilon$, $c$ to be specified later (as a function of $\alpha,n$)
\Statex

\State $C_{1,v}\gets\emptyset$ 

\For{$i=1,2,\ldots,\Theta(\log n)$}
  \State $k_i \gets \left\lceil (1-\epsilon)^{\,i-1}\, m \right\rceil$
  \If{$1 \le k_i \le 2B$}
     \State $T_{i,v}\gets \AllTasks$ \Comment{give all tasks to all nodes}
     \label{line:TallTasks}
  \Else
     \State Let $A_1,\dots,A_m$ be a $(k_i,B,\epsilon)$-\CF, agreed upon by all nodes.\label{line:generateCF}
     \State $T_{i,v}\gets \{\, t\in\AllTasks \mid v\in A_t \,\}$ \Comment{$T_{i,v}$ depends only on $(i,v)$}
  \EndIf

   \State $U_{i,v}\gets T_{i,v}\setminus C_{i,v}$ \Comment{only tasks not yet known to be completed}
   \State For $2B\cdot R$ rounds, complete tasks in $U_{i,v}$ by lexicographic order; idle if no tasks to complete.
     \label{line:doTasks}  
     \If {$v$ has completed all tasks in $U_{i,v}$}\label{line:condition_for_s}
     \State $v$ broadcasts  $s_{i,v} = 1$ to all nodes.
     \label{line:sendSuccess}
     \Else 
     \State $v$ broadcasts $s_{i,v}=0$.
     \label{line:sendNotSuccess}
     \EndIf 
     \State Receive $s_{i,u}$ from all nodes $u$. 
     \label{line:recvSuccess}\\ \Comment {$s_{i,u}=\bot$ if $u$ crashed before sending  to~$v$}

     \State 
     $\displaystyle C_{i+1,v}\gets C_{i,v}\  \cup\ \bigcup_{{\substack{u\in[n]:\\ s_{i,u}= 1}}}T_{i,u}$ 
     \label{line:updateC}
\EndFor
\end{algorithmic}
\end{algorithm}

Our task-completion algorithm, \Cref{alg:task-completion}, consists of $\Theta(\log n)$ iterations. In iteration~$i$, each node~$v$ is assigned a set of tasks $T_{i,v} \subseteq [m]$, deterministically generated via a \CF as a function only of $i$ and~$v$; these sets are therefore globally known to all nodes (we explain the construction in detail below).
The goal of node~$v$ is to complete all incomplete tasks in~$T_{i,v}$. Throughout the algorithm, each node maintains a list $C_{i,v}$ of tasks that~$v$ knows to be completed at the start of iteration~$i$, either because $v$ completed them in a previous iteration or because some other node completed them and reported this to~$v$.
In iteration~$i$, node~$v$ attempts to complete all tasks in $T_{i,v} \setminus C_{i,v}$, one after the other. 
Our analysis shows that, for all but a small fraction of the nodes, this set of potentially incomplete tasks has size at most~$2B$, for a parameter $B = B(\epsilon)$. 
Hence, after $2B \cdot R$ rounds, most nodes succeed in completing \emph{all} of their assigned tasks unless they crash during these $2B R$ rounds.
If node~$v$ does not crash during these rounds, and $|T_{i,v} \setminus C_{i,v}| < 2B$, then, upon completing all these tasks, $v$ broadcasts a success bit $s_{i,v} = 1$; otherwise, it broadcasts $s_{i,v} = 0$ to indicate that it did not complete all of its assigned tasks.

Initially, $C_{1,v} = \emptyset$. 
At the end of iteration~$i$, node~$v$ updates $C_{i+1,v}$ to include all tasks in $C_{i,v}$, as well as all tasks in~$T_{i,u}$ for every $u \in [n]$ such that $s_{i,u} = 1$. To do so, node~$v$ must know the sets $T_{i,u}$ for all $u \in [n]$, which indeed holds, as we explain next.

In iteration~$i$, we set $k_i = \lceil (1-\epsilon)^{i-1} m \rceil$. If $1 \leq k_i \le 2B$, all nodes set $T_{i,v} = [m]$. Otherwise, each node deterministically constructs a $(k_i,B,\epsilon)$-\CF using \Cref{lem:uniform_covering} with parameters $n = n$, $m = m$, $k = k_i$, and $\epsilon = \epsilon$ (the left-hand side corresponds to the lemma’s parameters, and the right-hand side to the task-completion parameters). This construction yields a family of sets $A_1,\dots,A_m \subseteq [n]$, where each set~$A_t$ consists of all nodes assigned to task~$t$. 
In other words, if $v \in A_t$, then $t \in T_{i,v}$. 
Since this construction is deterministic and depends only on $k_i$, $B$, and~$\epsilon$, all nodes share the same \CF and can therefore compute the sets~$T_{i,v}$ for all $v \in [n]$.

The properties of the \CF guarantee that almost all nodes are  assigned with $O(B)$ tasks that still need completion, and that each task that needs completion is assigned to $\approx nB/k_i$ different nodes. These two properties guarantee that enough progress is made in each iteration, namely, at least $\epsilon$-fraction of the tasks that need completion will be completed and all non-crashed nodes will be aware of this. After $O_\epsilon(\log n)$ iterations, all tasks will be completed.

\subsection{Analysis}
We now prove both the correctness of our algorithm and its round complexity.
For a task $t \in [m]$, we say that $t$ is \emph{fully verified} at the start of iteration~$i$ if $t \in C_{i,v}$ for all nodes that have not crashed by that time. 
Let $N_i$ denote the number of tasks that are not fully verified at the start of iteration~$i$. The following is immediate from the definition.
\begin{claim}
\label{clm:CisLarge}
     For any iteration~$i$ and any non-crashed~$v$, 
     $|C_{i,v}|\ge m-N_i$
\end{claim}
\begin{proof}
If there are $N_i$ tasks that are not fully verified, the other $m-N_i$ tasks are fully verified and in particular reside in $C_{i,v}$ for any non-crashed~$v$. 
\end{proof}

The main part of this section is the proof of the following lemma.
\begin{lemma} \ 
\label{lem:task_analysis}
\begin{enumerate}
    \item \label{lem:task_analysis_item_1} When \Cref{alg:task-completion} terminates, all tasks are fully verified.
    \item \label{lem:task_analysis_item_2} 
    \Cref{alg:task-completion} terminates after $O(R\cdot B\log n)$ rounds.
\end{enumerate}
\end{lemma}

We start by proving the first part of \Cref{lem:task_analysis}. Towards this goal we show that the variable~$C_{i,v}$ that $v$ holds indeed indicates only completed tasks. We then show that the number~$N_i$ of non-fully-verified tasks decreases to zero as the algorithm progresses. Once we reach an iteration~$i$ for which $N_i=0$, all tasks reside in all the $C_{i,v}$'s of the non-crashed nodes, which in particular means that all tasks were completed.

\begin{lemma}
\label{lem:verified_is_complete}
    If $t \in C_{i,v}$ for some $v$, then the task~$t$ has been completed before  start of iteration~$i$. 
\end{lemma}
\begin{proof}
    Proof by induction on $i$. 
    For $i=1$, the claim holds vacuously, since $C_{i,v} = \emptyset$.
    
    Now, assume the claim holds at the start of some iteration $i \geq 1$, and consider the start of the $(i+1)$-st iteration. 
    Recall that $v$ sets its new $C_{i+1,v}$ by
    \[
    C_{i+1,v}\gets C_{i,v}\ \cup\ \bigcup_{u\in[n]:\ s_{i,u}=1} T_{i,u}\text{.}
    \]

Consider a task $t \in C_{i+1,v}$. If $t \in C_{i,v}$, then by the induction hypothesis it was completed before iteration~$i$.
Next, consider a task $t \in \bigcup_{u \in [n]:\, s_{i,u} = 1} T_{i,u}$. In this case, there exists some node~$u$ such that $s_{i,u} = 1$ and $t \in T_{i,u}$. By the condition in line~\ref{line:condition_for_s}, node~$u$ has completed all tasks in $T_{i,u} \setminus C_{i,u}$. Consequently, either $t$ was completed by~$u$ during  iteration~$i$, 
or $t \in C_{i,u}$, in which case it was completed before iteration~$i$ by the induction hypothesis.
We therefore conclude that every task in~$C_{i+1,v}$ has been completed by the start of iteration~$i+1$.  
\end{proof}

Next, we show that for every iteration~$i$, the value~$k_i$ used by the algorithm upper-bounds the number of tasks that are not fully verified, namely~$N_i$.
\begin{lemma}
\label{lem:task_decreasing_ni}
    For any iteration $i$, $N_i \leq k_i$.
\end{lemma}
\begin{proof}
The lemma follows (by induction) directly from the following two claims, that hold for any iteration~$i$:
(a) $N_{i+1} \leq N_i \leq m$, and (b) if $(1-\epsilon) k_i \leq N_i \leq k_i$, then $N_{i+1} \leq (1-\epsilon) N_i$.
Indeed, for $i=1$ we have $N_1=k_1=m$. For the induction step, assume that  $N_i\le k_i$ and recall that 
$k_{i+1}=\lceil(1-\epsilon)^{i} m\rceil$.
If $N_i\le k_{i+1}$, the lemma holds from claim~(a); otherwise, the conditions of claim~(b) hold, and it follows that 
$N_{i+1}\le \lfloor (1-\epsilon)N_i\rfloor \le \lfloor(1-\epsilon)k_i \rfloor= \lfloor (1-\epsilon)\lceil(1-\epsilon)^{i-1}m\rceil\rfloor \le \lceil(1-\epsilon)^{i}m\rceil=k_{i+1}$, where the first inequality holds since $N_{i+1}$ is an integer, hence flooring the right-hand-side is allowed, and the last inequality holds for any $\epsilon\in(0,1)$ and integers~$i,m$ (see \Cref{lem:ceiling}). Let us now prove the above two claims.

The proof of (a) is immediate:
the number of tasks that are not fully verified is clearly bounded by the total number of tasks~$m$, and monotonically decreases as tasks are being completed.
   We continue to proving item~(b); let us split the proof into two sub-cases.
   
   If $k_i \le 2B$, then
   (b) follows since, in this case, all non-crashed nodes $v$ are assigned with all the tasks~$[m]$. 
   Consider any node $v'$ that is still non-crashed at the end of iteration~$i$.
   Note that $|C_{i,v'}|\ge m-N_i\ge m-k_i$, which follows from \Cref{clm:CisLarge} and from the conditions of part (b), i.e., $N_i \le k_i$. 
   During this iteration, $v'$ sets $T_{i,v'}=[m]$, and $U_{i,v'}=[m]\setminus C_{i,v'}$. By the above, $|U_{i,v'}|\le k_i \le 2B$, hence in line~\ref{line:doTasks}, $v'$ succeeds in completing \emph{all} the tasks in $U_{i,v'}$ and sets $C_{i+1,v'}=[m]$. Since the above applies to all non-crashed nodes, this makes all the tasks fully verified and $N_{i+1}=0$. 
   
   Let us now consider the case where  $k_i > 2B$.
    Let $I\subseteq[m]$ be the set of the $N_i$ tasks that are not fully verified at the beginning of iteration~$i$, and let $I'\subseteq[m]$ be the (lexicographically minimal) set of $|I'|=k_i$ tasks such that $I\subseteq I'$.
    %
    Let $A_1,\ldots, A_m$ be the $(k_i,B,\epsilon)$-\CF from line~\ref{line:generateCF}, and let $J_i$ be the set guaranteed by \Cref{def:covering_family}(\ref{item:covering_family_2}) for the $k_i$ subsets $\{A_{i'}\}_{i'\in I'}$. 

For any task~$t$ that is not fully verified at the start of iteration~$i$, we say that a non-crashed node~$v$ \emph{attempts to fully verify}~$t$ in iteration~$i$ if $t \in T_{i,v}$.
We claim that for any such task~$t$, if there exists a node~$v \in J_i$ that attempts to fully verify~$t$ and does not crash during iteration~$i$, then $t$ becomes fully verified by the end of the iteration.
To see this, fix such a node~$v$. Any task $t' \notin I'$ is already fully verified and therefore belongs to~$C_{i,v}$. Recall that $\load(j,I') \triangleq |\{\, i' \in I' \mid j \in A_{i'} \,\}|$, and that the \CF guarantees that for any $j\in J_i$,  $\load(j,I') \le (1+\epsilon)B \le 2B$ (\Cref{def:covering_family}(\ref{item:covering_family_2})). It follows that $|U_{i,v}| = |T_{i,v} \setminus C_{i,v}| \le \load(v,I') \le 2B$.
Hence, node~$v$ can complete all tasks in~$U_{i,v}$ within $2BR$ rounds (line~\ref{line:doTasks}). Since $t \in T_{i,v}$, task~$t$ is completed during this phase (if wasn't already complete). Node~$v$ then broadcasts $s_{i,v} = 1$, which makes~$t$ fully verified once this message is received by all non-crashed nodes.

To complete the proof of claim~(b), we need to show that for at least $\epsilon k_i$ tasks that are not fully verified, there exists a node in~$J_i$ that does not crash during iteration~$i$ and attempts (and thus succeeds) to complete them.
   We now count the total number of attempts made by all nodes in $J_i$ to verify some non-fully-verified task, i.e.,
   \[
   \#\mathrm{attempts} :=\sum_{\substack{t\text{ is not fully verified} \\ \text{at the start of iteration~$i$}}} |\{v \in J_i \mid t \in T_{i,v}\}|
   \text{.}
   \] 
   
   We first prove that \begin{equation}
   \label{eq:atemptes_lower_bound}
       \#\mathrm{attempts}  \geq (1-4\epsilon)nB.
   \end{equation} 
   Recall that $(1-\epsilon)k_i \leq N_i \leq k_i$, and that 
$I'\subseteq[m]$ contains exactly $k_i$ tasks that includes all the $N_i$ non-fully-verified tasks. Note that the number of fully verified tasks in $I'$ never exceeds~$\epsilon k_i$.
   By \Cref{def:covering_family}(\ref{item:covering_family_1}), the total number of nodes that attempt to verify in iteration~$i$ a task $t \in I'$ that is already fully verified  is at most $\epsilon k_i \cdot 2nB / k_i$; this bound also applies to their number within~$J_i$.
   Further, each node in $J_i$ is assigned with at most $(1+\epsilon)B$ different tasks, so all these attempts complete during the $2BR$ rounds of line~\ref{line:doTasks}.
   It follows that nodes in~$J_i$ make at least the following number of attempts to verify tasks that are not fully verified:
   \[
   \#\mathrm{attempts} \geq |J_i|(1-\epsilon)B - 2\epsilon nB \geq ((1-\epsilon)^2-2\epsilon)nB\text{,}
   \]
   and \Cref{eq:atemptes_lower_bound} follows.
   
   Moreover, the set~$J_i$ contains at most $\alpha n$ crashed nodes (including nodes that crash during iteration~$i$), which correspond to at most $(1+\epsilon)\alpha n B$ attempts that were counted in \Cref{eq:atemptes_lower_bound} but may not actually occur. Summing up, at least $(1 - 4\epsilon - (1+\epsilon)\alpha)nB$ attempts are made by non-crashed nodes in~$J_i$ to verify tasks that are not fully verified, and these attempts succeed in iteration~$i$.

Since each non-fully-verified task is associated with at most $2nB/k_i$ nodes that may attempt to verify it
(\Cref{def:covering_family}(\ref{item:covering_family_1})), a pigeonhole argument implies that the number of non-fully-verified tasks at the start of iteration~$i$ that become fully verified by the end of the iteration satisfies
\[
N_i - N_{i+1} \geq \frac{(1-4\epsilon-(1+\epsilon)\alpha)nB}{2nB/k_i}
= \frac{1-4\epsilon-(1+\epsilon)\alpha}{2}\,k_i\text{.}
\]
We choose $\epsilon > 0$ to be a sufficiently small constant such that $(1-4\epsilon-(1+\epsilon)\alpha)/2 \geq \epsilon$. It follows that at least $\epsilon k_i$ tasks from $N_i$ become fully verified by the end of iteration~$i$. Recalling that $N_i \le k_i$, it follows that at least an $\epsilon$-fraction of the non-fully-verified tasks become fully verified, 
\(
N_{i+1} \le N_i - \epsilon k_i \le (1-\epsilon)N_i,
\)
as required to complete the proof of claim~(b).
\end{proof}
The above key lemma now allows us to complete the proof of \Cref{lem:task_analysis}.
\paragraph*{Proof of \Cref{lem:task_analysis}:} By \Cref{lem:task_decreasing_ni}, when $k_i < 1$, then all tasks are fully verified; by setting the main loop to run $\Theta(\log n)$ until $k_i<1$, we guarantee that \Cref{lem:task_analysis}(\ref{lem:task_analysis_item_1}) holds at the end of the last iteration.\footnote{In fact, 
if $k_i\le 2B$, then at the end of this iteration, $N_{i+1}=0$ (see the proof of the first sub-case of claim~(b) in \Cref{lem:task_decreasing_ni}), and the algorithm terminates in this iteration. 
Still, $\Theta(\log n)$ iterations are needed overall. } The above argument also proves \Cref{lem:task_analysis}(\ref{lem:task_analysis_item_2}):
\Cref{alg:task-completion} terminates after $\Theta(\log n)$ iterations, where each iteration takes $2BR+1$ rounds (lines~\ref{line:doTasks}--\ref{line:recvSuccess}), totaling $\Theta(R\cdot B\log n)$ rounds.

\subsection{Completing the Proof of \texorpdfstring{\Cref{thm:main_task_complete}}{the Main Task-Completion Theorem}}
\label{sec:completingTaskThm}
By \Cref{lem:task_analysis},  \Cref{alg:task-completion} correctly solves the $m$-task completion problem with $m < cn$ tasks in $\Theta(R \cdot B \log n)$ rounds, despite up to $\alpha n$ crashes.
To solve the $M$-task completion problem, \Cref{alg:task-completion} is executed $\lceil M/m \rceil$ times on batches of $m = cn$ tasks each, where $c$ is the parameter given by \Cref{lem:uniform_covering} to obtain a $(k,B,\epsilon)$-\CF.
By choosing $B = \Theta(1/\epsilon^3)$, \Cref{lem:uniform_covering} guarantees the existence of a $(k_i,B,\epsilon)$-\CF for any $k_i \le m$ and any $B \le k_i$. 
Under this choice, the $M$-task completion problem is solved using $\lceil M/m \rceil = O(\lceil M/n\rceil)$ sequential invocations of \Cref{alg:task-completion}, where each instance takes $\Theta(R \cdot B \log n) = \Theta(R \log n)$ rounds.
This completes the proof of our main theorem, \Cref{thm:main_task_complete}.

\section{Robust Simulation of Congested Clique Algorithms}
\label{sec:simulation}
In this section, we show how to simulate any deterministic $T$-round Congested Clique algorithm~$\ALG$,  in the presence of up to~$\alpha n$ crashes.  
In particular, we prove the following.
\begin{theorem}
    \label{thm:main_simulation}
Let $\ALG$ be a deterministic congested-clique algorithm that runs for $T$ rounds, where the input of each node has size $O(n \log n)$. 
Then, for any $\alpha \in [0,1)$, there exists a congested-clique algorithm that simulates the computation of~$\ALG$ in the presence of up to $\alpha n$ crashes and runs in $O(T^2 \log n + T \log^2 n)$ rounds.
\end{theorem}

Before describing the algorithm, we briefly introduce and define several primitives and building blocks used in our simulation and establish some useful notation.

\subsection{Storing and Retrieving Information in the Network}
\label{sec:store_retrieve}
In this section, we describe several basic procedures that serve as building blocks for our simulation, namely $\Store$ and $\Retrieve$, which are used to robustly maintain information in the network.

\paragraph*{Error correction codes.} 
To store and retrieve information from the network we utilize an error correction code (see \Cref{sec:ECC}). 
We remark that any code with constant rate and relative distance strictly greater than~$\alpha$ suffices for our purposes, including codes over constant-size alphabets.
However, for simplicity,
we will use the MDS code given by \Cref{fct:ErasureECC} with $N=n$ and  $K=\rho n$ (an integer) for 
some constant rate $\rho < (1-\alpha)$ that implies relative  
distance $\delta > \alpha +1/n$. By \Cref{fct:ErasureECC}, the code can correctly decode any codeword that suffers up to~$\alpha n$ erasures.
To further simplify the usage of this code, we will assume $n$ is a prime and use the alphabet~$\F_n$; if $n$ is not prime, a larger field can be used instead.
Hence, we can write $\Enc:[n]^{\rho n} \rightarrow [n]^n$ and $\Dec:[n]^n \rightarrow [n]^{\rho n}$ as the encoding and decoding algorithms of the error correction code, respectively. 

\paragraph*{Storing information in the network.} 
The $\Store$ procedure saves information in the network so that the information is recoverable even when up to $\alpha n$ nodes crash. 
The operation $\Store(S,I)$ is executed by a single node and takes two inputs: a string $S \in [n]^*$ and a key $I \in [\poly(n)]$. The string $S$ represents the information to be stored, while the key~$I$ serves as an index that enables referring to and retrieving this information at a later time.
To store $S$ in the network, the node 
splits $S$ into parts of size $\rho n$ each, padding the last part with zeros if needed; denote these parts by $S_1,\dots,S_{\lceil |S|/\rho n\rceil}$. 
The node then encodes each part $S_j$ using an error correction code to obtain the $n$-symbol codeword $\Enc(S_j)$. 
The node then sends one symbol of each codeword to each node, namely, for $\ell\in[n]$, the node sends $\Enc(S_j)[\ell]$ to node~$\ell$, along with the metadata $I$, and~$j$. 
Once the node~$\ell$ receives this information it saves it locally in a variable, which we denote  
$\mathrm{Stored}^\ell_{I,j}$.

\begin{observation}
\label{obs:StoreComplexity}
    For $S\in[n]^*$, $\Store(S,\cdot)$ takes $O(\lceil |S|/n \rceil)$ rounds.
\end{observation}
For the purposes of analysis, a $\Store$ operation is said to be \emph{successful} if the executing node~$v$ does not crash before the end of the procedure.
We also observe that the operation $\Store(S,I)$ is idempotent. Specifically, if multiple nodes hold the same pair $(S,I)$ and execute $\Store(S,I)$, the information sent to each node in the network is identical. Consequently, the local variable $\mathrm{Stored}^{\ell}_{I,j}$ stored by node~$\ell$ attains the same value regardless of which node sent it. 
Thus, node~$\ell$ needs to store only a single copy of $\mathrm{Stored}^{\ell}_{I,j}$, and all stored values are consistent with the execution of $\Store(S,I)$, independent of the sender.

\paragraph*{Retrieving stored information from the network.} 
The $\Retrieve$ procedure takes as input a key~$I$ and aims to retrieve the string~$S$ that was previously stored in the network using this key.
Toward that goal, the node~$v$ that executes $\Retrieve(I)$ sends $I$ to all other nodes in the network. In return, each node $\ell\in[n]$ replies with a tuple $(j,I,\mathrm{Stored}^\ell_{I,j})$ for each $j$ it holds.  
For each such $j$, let $C_j$ be the string defined by 
\begin{equation}
    C_j[\ell] =
    \begin{cases*}
        c_{I,j}
         & if $(j,I,c_{I,j})$ was received from $\ell$  \\
      \bot        & otherwise
    \end{cases*}
  \end{equation}

Node~$v$ then computes $S_j = \Dec(C_j)$ for all indices~$j$ for which $C_j$ is non-empty, and outputs the concatenated string $S = (S_1, S_2, \ldots)$.

\begin{observation}
\label{obs:retrieveSucceeds}
Let $S\in[n]^*$ be a string that was successfully stored using key~$I$. 
Then the procedure $\Retrieve(I)$ completes in $O(\lceil |S|/n \rceil)$ rounds. 
Moreover, if the node executing $\Retrieve(I)$ does not crash during this procedure, it outputs~$S$.
\end{observation}
The above holds because, if $S$ was successfully stored in the network, each non-crashed node holds the corresponding symbol $\Enc(S_j)$ for every~$j$. Consequently, each codeword~$C_j$ obtained by the $\Retrieve$ procedure is missing at most $\alpha n$ symbols due to node crashes. Therefore, the decoder successfully reconstructs the original value, that is, $\Dec(C_j) = S_j$.

\paragraph*{Inputs, outputs, states, and bookkeeping.} 
During the simulation algorithm described in the next section, the network stores and retrieves multiple strings corresponding to nodes’ states and messages exchanged during the $T$ rounds of the simulated algorithm~$\ALG$. 
In particular, for each round~$r$ of~$\ALG$, we store and retrieve the string~$S_j(r)$, which contains all the $n$ messages received by node~$j$ in round~$r$ of~$\ALG$, and the string~$M_j(r)$, which contains the $n$ messages sent by node~$j$ in round~$r$ of~$\ALG$. 
The string~$S_j(0)$ denotes the input of node~$j$ to the algorithm~$\ALG$; as specified in the model, these inputs are already stored in the network when the simulation begins (see \Cref{sec:prelim-Network}).

Since $\ALG$ is a congested-clique algorithm, each entry of~$S_j(r)$ and~$M_j(r)$ is a message of size $O(\log n)$ bits. 
Specifically, $S_j(r)[i]$ denotes the message that node~$j$ receives from node~$i$ in round~$r$, while $M_j(r)[i]$ denotes the message that node~$j$ sends to node~$i$ in round~$r$.
Without loss of generality, we assume that the output of~$\ALG$ is all messages sent in round~$T$. Equivalently, for each node~$j$, the string $M_j(T)$ represents its output. 
Hence, a correct simulation of~$\ALG$ reduces to computing $M_j(T)$ and performing $\Store(M_j(T))$ for all $j \in [n]$.

To simplify the notation and the description of the simulation, we assume a fixed numbering of the $2n(T+1)$ strings $\{S_j(r), M_j(r)\}$. Thus, for each $j \in [n]$ and each $r \in \{0, \ldots, T\}$, the key~$I$ corresponding to $S_j(r)$ or $M_j(r)$ is fixed and known to all nodes, and we therefore omit it from the simulation description.

\subsection{Simulation Algorithm}
\label{sec:simulationAlgDescription}
Let us overview our task-completion based crash-resilient simulation of a $T$ round (non crash resilient) congested-clique algorithm~$\ALG$.
The simulation proceeds in $T$ iterations, where iteration~$r$ simulates the $r$-th round of~$\ALG$. 
In ``simulating'' the $r$-th round, we mean that the simulation stores in the network the state each  node $j\in[n]$ holds in the end of the $r$-th round of~$\ALG$.
The \emph{state} of the node~$j$ (at the end of round~$r\ge1$) contains its input and all messages it has received  until the end of round~$r$.\footnote{We remark that one can redefine the state of a node in round~$i$ to be the local memory held by the node at that round. While this memory can always be derived from the node’s inputs and incoming messages, in some cases it may be significantly shorter. 
Under this definition, the resulting complexity of our simulation is $O\!\left(T \cdot \bigl(\lceil s/(n \log n) \rceil + \log n\bigr)\log n\right)$, where $s$ denotes the size of the state. 
In the general case, where the state includes all inputs and incoming messages, we have $s \le T\cdot n \log n$, which yields the stated complexity.}
We recall that the $j$-th input, $S_j(0)$, is already stored in the network at the onset of the simulation. Hence the state of $j$ at the end of round~$r$ is the tuple $(S_j(0),\ldots, S_j(r))$.

Somewhat informally, to simulate the $r$-th round, assuming that the $(r-1)$-st round has been successfully simulated, we perform following steps, carried out in parallel for each $j \in [n]$: 
(1) retrieving the state of node~$j$, which allows the simulation to compute all messages sent by~$j$ in round~$r$ of~$\ALG$; 
(2) computing the messages that~$j$ receives in round~$r$ of~$\ALG$; and 
(3) storing the updated state of~$j$ back in the network.
While this outline is conceptually straightforward, its implementation in the presence of crashes is considerably more complex: node~$j$ may have crashed and therefore must be simulated by another node, which itself may also crash, and so on. 
To address this challenge, we employ the task-completion framework to carry out each of the above steps.

\begin{enumerate}
\item To compute all the outgoing  messages of node~$j$ in round~$r$ of~$\ALG$, we define a \emph{task completion instance} in which completing the $\ell$-th task consists of the following steps:  
(1) retrieving the state of node~$\ell$ at the end of round~$r-1$, that is, retrieving the strings $S_\ell(0), \ldots, S_\ell(r-1)$;  
(2) using this state to locally compute $M_\ell(r)$, the set of all messages that node~$\ell$ sends in round~$r$ of~$\ALG$; and  
(3) storing  $M_\ell(r)$ in the network (see \Cref{task:ComputeOutgoingMessages}).

\item The next step is to deliver to each node~$j$ all messages designated to it in round~$r$, namely the vector $(M_1(r)[j], \ldots, M_n(r)[j])$, assuming that each string $M_i(r)$ has been stored in the network. To this end, we employ two nested task-completion instances.  

The \emph{inner} task (\Cref{task:inTask}) is responsible for retrieving the stored $M_j(r)$ and forwarding its constituent messages $M_j(r)[1], \ldots, M_j(r)[n]$ to the nodes simulating nodes $1$ through $n$, respectively. 
The identity simulated by each non-crashed node 
is determined by the \emph{outer} task (\Cref{task:outTask}). 
In this task, the node simulating identity~$\ell$ constructs $S_\ell(r)$ by collecting the set of messages received by~$\ell$ in round~$r$ of~$\ALG$ sent during the inner task.

\item Once all incoming messages of a simulated node~$\ell$ have been collected, the simulating node stores $S_\ell(r)$ in the network. This storage step is performed at the end of the outer task (\Cref{task:outTask}).

\end{enumerate}
The complete simulation is depicted in \Cref{alg:simulation}.
\begin{algorithm}[ht]
\caption{Crash-resilient Simulation Algorithm}
\label{alg:simulation}
\begin{algorithmic}[1]
\Statex \textbf{Input:} A $T$-round Algorithm~$\ALG$ to simulate. Parameter $\alpha\in(0,1)$; An error correction code stated in \Cref{sec:store_retrieve}.
\Statex
\For{$r = 1,\dots,T$}
\State Execute \Cref{alg:task-completion} on \Cref{task:ComputeOutgoingMessages} with input~$r$.   \label{line:sim:Run-computeMessages}
\State Execute \Cref{alg:task-completion} on \Cref{task:outTask} with input~$r$.
\label{line:sim:Run-outerTask}
\EndFor

\end{algorithmic}
\end{algorithm}

\begin{TaskInstance}[H]
        \caption{The $\ComputeOutgoingMessages$ Task}
        \label{task:ComputeOutgoingMessages}
        \begin{algorithmic}[1]
        \Statex \textbf{Input:} the simulated round~$r$.
        \Statex
        \Statex \textbf{To complete task} $\ell \in [n]$:
        \For {$i\in\{0,1,\ldots, r-1\}$}                 
        \State $\Retrieve(S_\ell(i))$ 
        \EndFor
        \State As a function of $S_\ell(0),\ldots,S_\ell(r-1)$, 
        compute $M_\ell(r)$: the messages node $\ell$ sends in round~$r$ of $\ALG$. 
        \State $\Store(M_\ell(r))$.
        \Statex
        \Statex \raisebox{.1ex}{$\blacktriangleright$} The task is complete once $M_\ell(r)$ is successfully stored.
        \end{algorithmic}
    \end{TaskInstance}
    
\begin{TaskInstance}[H]
\caption{The $\inTask$}
\label{task:inTask}
\begin{algorithmic}[1]
\Statex \textbf{Input:} a vector $(\ell_1,\ldots,\ell_n)$, the simulated round~$r$.
\Statex
\Statex \textbf{To complete task $\ell\in[n]$:}
\State $\Retrieve (M_\ell(r))$ 
\label{line:inner:RetrieveM}
\For{each $j \in [n]$ for which $\ell_j \ne \bot$}
    \State Send $(\ell,M_{\ell}(r)[\ell_j])$ to node~$j$.
    \Comment{$j$ simulates $\ell_j$.} \label{line:inner:sendM}
\EndFor

\Statex
\Statex \raisebox{.1ex}{$\blacktriangleright$} The task is complete once every non-crashed node~$j$ has received $(\ell, M_{\ell}(r)[\ell_j])$.
\end{algorithmic} 
\end{TaskInstance}

\begin{TaskInstance}[H]
\caption{The $\outTask$}
\label{task:outTask}
\begin{algorithmic}[1]
\State
\textbf{Input:} the simulated round~$r$.
\Statex
\Statex \textbf{To complete task} $\ell \in [n]$:
\State Broadcast $\ell$ to all nodes; receive $\ell_j$ from node~$j\in [n]$ or set $\ell_j = \bot$ if no value is received.
\label{line:outtask:set-ell}
\State Execute \Cref{alg:task-completion} on \Cref{task:inTask} (the inner task) with input $(\ell_1, \ldots, \ell_n),r$. \label{line:out_call_in}
\State For any value $(j, \mathrm{val}_j)$ received during the execution of the inner task, set $S_\ell(r)[j] \gets \mathrm{val}_j$ \label{line:out_set_values}
\State $\Store(S_\ell(r))$. 
\label{line:outer:storeS}
\Comment{The inner task guarantees that all indices of $S_\ell(r)$ are set} 
\Statex
\Statex \raisebox{.1ex}{$\blacktriangleright$} The task is complete once $S_\ell(r)$ is successfully stored.
\end{algorithmic}
\end{TaskInstance}

\subsection{Analysis}
\label{sec:simulation_Analysis}
Towards proving \Cref{thm:main_simulation} we now analyze the correctness and complexity of \Cref{alg:simulation}.
Our first step is showing that  all the task completion instances we define are both idempotent,  independent, and \leadersafe (see \Cref{sec:task_completion_def}), so we could execute our task-completion \Cref{alg:task-completion} on each of them.
%
Then, we show by induction that at the end of iteration $r$ of the simulation, $S_\ell(r)$ and $M_\ell(r)$ are stored in the network, for all~$\ell\in[n]$. 
This implies the correct simulation of~$\ALG$, as all messages it sends are correctly computed and communicated, and its output $\{M_j(T)\}_{j\in[n]}$ is stored in the network.

\begin{lemma}
    The tasks defined as part of 
    Task Completion Instances 1--3 
    are independent, idempotent, and \leadersafe.
\end{lemma}
\begin{proof}
    The independence of the inner task (\Cref{task:inTask}) and
    the compute-message task (\Cref{task:ComputeOutgoingMessages}) 
    easily follows from the fact that
    these tasks only retrieve information that is already stored in the network, and stores new information not in use by the other tasks. Moreover, no congestion arises even when all $n$ tasks are executed simultaneously, since each node communicates with each of its neighbors at most once per round.

    For the independence of the outer task (\Cref{task:outTask}), again the inputs only depend on $r$ and the task itself, and information stored in one task is not an input of any other outer task, nor the any of the inner sub-tasks. 
    Moreover, all tasks can be executed simultaneously without causing congestion. To see this, note that in all lines except for the call to the inner task in line~\ref{line:out_call_in}, each node communicates with each of its neighbors at most once per round.
    Line~\ref{line:out_call_in} initiates invocation of the task-completion \Cref{alg:task-completion} on the inner task instance, which is independent as we proved above.

    Let us now consider the idempotence of the three task instances.
    The idempotence of the outer task
    and the compute-messages task
    follows from the fact that the $\Store$ procedure is idempotent conditioned on all the nodes storing the same information. 
    For the inner task instance, note that completing the $\ell$-th inner task is defined as the event that all non-crashed nodes~$j$ receive the message $(j, M_\ell(r)[\ell_j])$, rather than as the act of a specific node sending these messages. 
    The reason is that different nodes may have different input vectors $(\ell_1, \ldots, \ell_n)$, which may differ only at indices corresponding to crashed nodes (namely: some nodes will see a $\bot$ for a node that crashed while announcing its $\ell$ in line~\ref{line:outtask:set-ell} of \Cref{task:outTask}). 
    Consequently, some nodes may send the message $(j, M_\ell(r)[\ell_j])$ to a crashed node, while other nodes attempting to complete the $\ell$-th inner task may not. Nevertheless, with respect to the non-crashed nodes, the final outcome is identical. Hence, the inner tasks are idempotent.

    Finally, we argue that all three task instances are \leadersafe. The inner and compute-messages task instances are \leadersafe since the node attempting to complete a task executes the $\Retrieve$ operation by itself (which always succeeds in the presence of up to $\alpha n$ crashes; see \Cref{obs:retrieveSucceeds}) and subsequently performs the $\Store$ operation or sends the required messages on its own.
    The argument for the outer task instance is slightly more involved, as each outer task invokes an inner task completion instance, and thus the completion of an outer task depends on the correct completion of all the corresponding inner task, performed by other nodes. 
    However, by the correctness of \Cref{alg:task-completion} (\Cref{thm:main_task_complete}), 
    we are guaranteed that once line~\ref{line:out_call_in} completes, all elements of $S_\ell(r)$ have been delivered to the node~$v$ attempting to complete the outer task~$\ell$, \emph{regardless} of all crashes except for that of~$v$. This directly means that  outer tasks are  \leadersafe.
\end{proof}

By the above lemma, executing \Cref{alg:task-completion} on each task-completion instance successfully completes all tasks associated with that instance, even in the presence of up to $\alpha n$ crashes. 
We next show that our crash-resilient simulation (\Cref{alg:simulation}) correctly simulates any (non-resilient) algorithm~$\ALG$ on a round-by-round basis.

\begin{theorem}
\label{lem:sim_invariants}
For any iteration $1 \leq r \leq T$ of \Cref{alg:simulation},
    \begin{enumerate}
        \item \label{item:state_to_out} If $S_\ell(0),\dots,S_\ell(r-1)$ is successfully stored in the network for all $\ell \in [n]$ prior to the execution of \Cref{task:ComputeOutgoingMessages} (line~\ref{line:sim:Run-computeMessages} in \Cref{alg:simulation}), then at the end of this execution, $M_\ell(r)$ is successfully stored in the network for all $\ell \in [n]$.
        \item \label{item:in_to_state} If $S_\ell(0),\ldots,S_\ell(r-1)$ and $M_\ell(r)$ are stored for all $\ell \in [n]$ prior to the 
        execution of \Cref{task:outTask} (line~\ref{line:sim:Run-outerTask} in \Cref{alg:simulation}),
        then at the end of this execution, $S_\ell(r)$ is stored for all $\ell \in [n]$.
    \end{enumerate}
\end{theorem}
\begin{proof}
Item~(\ref{item:state_to_out}) follows from the fact that $M_\ell(r)$ is a function of $S_\ell(0), \ldots, S_\ell(r-1)$. The $\ell$-th task in \Cref{task:ComputeOutgoingMessages} consists of retrieving all strings $S_\ell(0), \ldots, S_\ell(r-1)$, locally computing $M_\ell(r)$, and storing the result in the network. Therefore, once the $\ell$-th task is completed, the statement follows. Note that, due to the error-correction code, the retrieval of $S_\ell(0), \ldots, S_\ell(r-1)$ succeeds regardless of which nodes have crashed, provided that the node executing the $\ell$-th task does not crash during its execution  (\Cref{obs:retrieveSucceeds}).

For Item~(\ref{item:in_to_state}), assume that some node~$i$ executes the $\ell$-th task of the outer task and does not crash before completing it. At the beginning of the $\ell$-th task, the node~$i$ broadcasts the value $\ell$ to the rest of the network, indicating that it assumes the role of simulating node~$\ell$ in~$\ALG$ for this iteration. Since node~$i$ does not crash during the execution of this task, this message is delivered to all non-crashed nodes, and each such node sets $\ell_i = \ell$ when executing line~\ref{line:outtask:set-ell} of \Cref{task:outTask}.
Recall that in the inner task (\Cref{task:inTask}), the $i$-th task consists of retrieving $M_i(r)$ (line~\ref{line:inner:RetrieveM}) and sending the pair $(i, M_i(r)[\ell_j])$ to node~$j$. By the correctness of the task completion protocol, all tasks of
the inner task
are completed by the time the nodes reach line~\ref{line:outer:storeS} of the outer task (\Cref{task:outTask}). Consequently, at that point, node~$i$ is guaranteed to have received the values $(j, M_j(r)[\ell])$ for every $j \in [n]$, which together constitute exactly the state~$S_\ell(r)$. 
Therefore, node~$i$ has all the information required to store $S_\ell(r)$ at line~\ref{line:outer:storeS}, and this storage operation succeeds since node~$i$ does not crash.

By the correctness of the task completion algorithm (\Cref{alg:task-completion}), all outer tasks are completed by the time the execution reaches line~\ref{line:sim:Run-outerTask} in \Cref{alg:simulation}, which completes the proof of this case.
\end{proof}

Next, we analyze the round complexity of the simulation.

\begin{theorem}
\label{thm:simComplexity}
    \Cref{alg:simulation} terminates after $O(T^2\log{n}+T\log^2{n})$ rounds.
\end{theorem}
\begin{proof}
The algorithm executes $T$ iterations,
where in each iteration we have one invocation of \Cref{alg:task-completion} on \Cref{task:ComputeOutgoingMessages} and one invocation on \Cref{task:outTask}.

Consider a compute-messages task for iteration~$r$. Completing this task requires performing $r$ $\Retrieve$ operations, each on a string of size $O(n \log n)$ bits (i.e., $n$ symbols), which together take $O(r) = O(T)$ rounds by \Cref{obs:retrieveSucceeds}. The task then performs a single $\Store$ operation on a string of size $O(n \log n)$ bits, which takes $O(1)$ rounds by \Cref{obs:StoreComplexity}. Overall, completing a single compute-messages task takes $R_{\text{comp-msg}} = O(T)$ rounds.
Therefore, by \Cref{thm:main_task_complete}, the round complexity of completing all $n$ compute-messages tasks in the presence of crashes is~$O(T \log n)$.

Next, consider the outer task. Broadcasting the identity~$\ell$ requires a single round. Invoking \Cref{alg:task-completion} on \Cref{task:inTask} takes $O(\log n)$ rounds by \Cref{thm:main_task_complete}; This holds since each of the $n$ inner tasks can be completed in $R_{\text{inner}} = O(1)$ rounds, as each such task consists of a single $\Retrieve$ operation on a string of size $O(n \log n)$ and a single round of message sending.
Finally, storing $S_\ell(r)$ takes $O(1)$ rounds. Overall, completing a single outer task takes $R_{\text{outer}} = O(\log n)$ rounds. Therefore, by \Cref{thm:main_task_complete}, the round complexity of completing all $n$ outer tasks using \Cref{alg:task-completion} is $O(\log^2 n)$.

Summing it all together, each iteration of \Cref{alg:simulation} takes $O(T\log n+ \log^2 n)$, hence its total round complexity is 
$O(T^2\log{n}+T\log^2{n})$.   
\end{proof}

\paragraph{Completing the Proof of \texorpdfstring{\Cref{thm:main_simulation}}{the Main Simulation Theorem}.}
The round complexity of \Cref{alg:simulation} is established in \Cref{thm:simComplexity}. Correctness follows by induction using \Cref{lem:sim_invariants}. Initially, the inputs $S_j(0)$ for all $j \in [n]$ are stored in the network. 
Assume that, for some round~$r \geq 1$, the strings $S_j(0), \ldots, S_j(r-1)$ are stored for all $j \in [n]$. Then, in the next iteration of \Cref{alg:simulation}, \Cref{lem:sim_invariants}(1) guarantees that the messages $M_j(r)$ are successfully computed and stored for all $j \in [n]$, and \Cref{lem:sim_invariants}(2) ensures that the corresponding  $S_j(r)$ are subsequently stored.
After $T$ iterations, all messages sent by~$\ALG$ have been simulated and stored. 
Recall that the output of~$\ALG$ is $\{M_j(T)\}_{j \in [n]}$ (\Cref{sec:prelim-Network}).
It follows from \Cref{lem:sim_invariants}(1) that this information is stored in the network at the end of the simulation. Therefore, the execution of~$\ALG$ has been correctly simulated, completing the proof.

\section*{Acknowledgments}
We thank Keren Censor-Hillel for helpful comments on an initial draft of this manuscript. Orr Fischer is supported in part by the Israel Science Foundation, grant No. 1042/22 and 800/22.


\bibliographystyle{plain}
\bibliography{bibliography}

\clearpage
\appendix


\section*{Appendix}
\section{Additional Technical Lemmas}

\begin{theorem}[Chernoff inequality for independent Bernoulli variables] \label{thm:chernoff}
Let $X_1,\ldots, X_n$ be mutually independent 0--1 random variables with $\Pr(X_i=1)=p_i$. Let $X=\sum_{i=1}^n X_i$ and set $\mu=\E[X]$. The following holds,
\begin{enumerate}
    \item 
    \label{chernoff:larger_side}
    for $0<\delta\le 1$,
    \(
    \Pr(X\ge (1+\delta)\mu) \le
        e^{-\mu\delta^2/2}
    \)
    \item 
    \label{chernoff:smaller_side}
    for $0<\delta< 1$,
    \(
    \Pr(X\le (1-\delta)\mu) \le
        e^{-\mu\delta^2/2}
    \)
    \item \label{item:chernoff_two_sided} for $0<\delta\le 1$,
    \(
    \Pr(|X-\mu| \ge \delta\mu) \le
        2e^{-\mu\delta^2/3}
    \)
    \item for $R \ge 6\mu$,
    \(
    \Pr(X\ge R) \le 2^{-R}
    \)
\end{enumerate}
\end{theorem}
For proof, see Theorems~4.4 and 4.5 in \cite{MU17book}.

\begin{lemma}
\label{lem:ceiling}
    Let $x \in [0,1)$ and $y \geq 0$. Then $\lfloor x \lceil y \rceil\rfloor \leq \lceil xy \rceil$.
\end{lemma}
\begin{proof}
    Since $y \geq \lceil y\rceil-1$, we have $xy \geq x(\lceil y \rceil-1)$. Taking ceiling on both sides, we get $\lceil xy \rceil \geq \lceil x\lceil y \rceil - x\rceil$. Let $k = \lfloor x \lceil y \rceil \rfloor$, and $\epsilon = x \lceil y \rceil - k$. We notice that $\epsilon \in [0,1)$. We have that $$\lceil x\lceil y \rceil - x\rceil = \lceil k+\epsilon-x \rceil \geq k = \lfloor x \lceil y \rceil \rfloor,$$
    where the first equality follows from the definition of $\epsilon$, the second from the fact that $\epsilon-x > -1$ (since $x \in [0,1)$), and that $k$ is an integer, and the final equality follows by definition of $k$. Combining all inequalities together, we get that
    $$\lceil xy \rceil \geq \lceil x\lceil y \rceil - x\rceil \geq k = \lfloor x \lceil y \rceil \rfloor,$$ 
    as required.  
\end{proof}


\section{A Lower Bound on the Round Complexity of Task Completion}
\label{app:lowerbound}
In this section, we show a very simple $\Omega(\log{n}/\log\log{n})$ lower bound on the round complexity of any task completion algorithm. This bound is essentially the same as the bound proven in \cite{GRS04,GS08} and is re-proven here for completeness.
\begin{theorem}
    Any deterministic
    algorithm solving the task completion problem with $M=n$ tasks and $R=1$ in the presence of $\alpha n$ crash-faults for $\alpha\in(0,1)$, 
    requires $\Omega(\log{n}/\log\log{n})$ rounds.
\end{theorem}

Fix an algorithm~$\ALG$ that solves the task completion problem. We construct an adversarial crash-pattern that guarantees that, at the end of round~$i$ of~$\ALG$, there are at least $\lfloor \alpha n / \log^i n \rfloor$ incomplete tasks, for any $i = O(\log{n}/\log\log{n})$. In particular, this implies that $\ALG$ requires  $\Omega(\log n / \log\log n)$ rounds to complete all $n$ tasks. We assume that $n$ is at least as large as some sufficiently large constant $C \leq n$, and for convenience halt when the number of incomplete tasks reduces below $C$. 
In particular, we assume that for all relevant values of~$i$, 
\[
\lfloor \alpha n/\log^{i}{n} \rfloor \leq  \lfloor \alpha n/\log^{i-1}{n} \rfloor/2.
\]

Consider the first round.
Since there are $n$ incomplete tasks and $n$~nodes, by an averaging argument, there exists a set $T_1 \subseteq [n]$ of size $|T_1| = \lfloor \alpha n/\log{n} \rfloor$ incomplete tasks for which at most two nodes are assigned for their completion. This follows  from the following claim:
\begin{claim}
\label{clm:tasksWithSingleNode}
    Assume that in a specific round there are $n'$ non-crashed nodes and  $k' \le n'$ incomplete tasks. 
    Then, there exists a set $T_{n',k'}\subseteq [k']$ of size at least $\lfloor k'/2 \rfloor$ 
    such that each task in $T_{n',k'}$ was attempted to be completed by at most $2n'/k'$ nodes at this round. 
\end{claim}
\begin{proof}
    Assume toward contradiction that at least $\lceil k'/2\rceil$ tasks were attempted to be completed by at least $2n'/k'+1$ nodes. Since each non-crashed node attempts at most one task, then we get a contradiction by
    \[
    n' \geq \lceil k'/2\rceil (2n'/k'+1) > n'.
    \]
    
\end{proof}

In the first round of~$\ALG$, the adversary crashes all nodes that attempt to complete tasks in a set $T_1 \subseteq T_{n,n}$, whose existence follows from the above claim. Since $|T_1| = \lfloor \alpha n / \log n \rfloor$, then $2|T_1|$ bounds the number of nodes crashed in this round. As a result of this crash pattern, all tasks in~$T_1$ remain incomplete at the end of the round. Indeed, for each task in~$T_1$, at most two nodes attempt to complete it, and these nodes are crashed by the adversary.

Next, in the second round of~$\ALG$, Claim~\ref{clm:tasksWithSingleNode} guarantees that there exists a set $T_2 \subseteq T_{n-2|T_1|,|T_1|} \subseteq T_1$ of size $|T_2| = \lfloor \alpha n/\log^2{n} \rfloor \leq |T_1|/2$, for which at most $2n/|T_1| \leq \frac{4\log{n}}{\alpha}$ nodes are assigned for their completion. The adversary crashes the set of nodes attempting to complete a task in $T_2$, resulting in at most $(2n/|T_1|)\cdot |T_2| \leq 4n/\log{n}$ crashes, and all tasks in the set $T_2$ being incomplete at the end of the second round.

More generally for $i > 1$, in round $i$, we assume that there is a set $T_{i-1}$ of size $|T_{i-1}| = \lfloor \alpha n/\log^{i-1}{n} \rfloor$ that is incomplete. By \Cref{clm:tasksWithSingleNode}, there exists a set $T_i \subseteq T_{i-1}$ of size $|T_i| = \lfloor \alpha n/\log^{i}{n} \rfloor \leq |T_{i-1}|/2$ tasks for which at most $2n/|T_{i-1}| \leq \frac{4\log^{i-1}{n}}{\alpha}$ nodes are assigned for their completion. The adversary crashes the set of nodes attempting to complete a task in $T_i$, resulting in at most $(2n/|T_{i-1}|)\cdot |T_i| \leq 4n/\log{n}$ crashes, and all tasks in the set $T_i$ being incomplete at the end of round $i$.

We get that $|T_i| < C$ when $i = \Omega(\log{n}/\log\log{n})$. Denote by~$r$ the first round where $|T_r| < C$; then $r = \Theta(\log n/\log\log n)$. Moreover, denoting $T_0 = [n]$, we get that the number of crashes in total is at most
$$
\sum_{i = 1}^{r} (2n/|T_{i-1}|) \cdot |T_i|
\leq  \sum_{i=1}^r \frac{4\log^{i-1}{n}}{\alpha} \Big\lfloor \alpha n/\log^{i}{n} \Big\rfloor
\leq  \sum_{i = 1}^{r}4 n/\log{n}
\leq \alpha n.
$$

\end{document}